\documentclass[a4paper,UKenglish]{notEurocg}
\usepackage{thmtools}
\usepackage{amsthm}
\usepackage{hyperref}
\usepackage[capitalise,nameinlink]{cleveref}
\usepackage{todonotes}
\newtheorem{lem}[theorem]{Lemma} %
\newtheorem{claim}{Claim}

\usepackage{paralist}

\newenvironment{claimproof}{\noindent\textsf{\bfseries\color{gray} Proof.}}{\hfill$\vartriangleleft$

\medskip}

\definecolor{defblue}{rgb}{0.121,0.47,0.705}
\let\emph\relax\DeclareTextFontCommand{\emph}{\color{defblue}\em}

\title{Grounded String Representations \\ of~Series-Parallel Graphs~without~Transitive~Edges\thanks{This research was initiated at the Bertinoro Workshop on Graph Drawing 2025. The second author acknowledges support of Czech Science Foundation by Research grant No. GA\v{C}R 23-04949X.}}

\titlerunning{Grounded String Representations 
  of~Series-Parallel~Graphs}

\author[1]{Sabine Cornelsen}
\affil[1]{University of Konstanz \\ \texttt{sabine.cornelsen@uni-konstanz.de}}
\author[2]{Jan Kratochv\'{\i}l}
\affil[2]{Faculty of Mathematics and Physics, Charles University \\
	\texttt{honza@kam.mff.cuni.cz}} %
\author[3]{Miriam M\"unch}
\affil[3]{University of Passau \\ \texttt{muenchm@fim.uni-passau.de}}%
\author[4]{Giacomo~Ortali}
\affil[4]{Università Telematica Internazionale Uninettuno \\ 	\texttt{giacomo.ortali@uninettunouniversity.net}}
\author[5]{Alexandra Weinberger}
\affil[5]{FernUniversit\"at in Hagen \\ %
\texttt{alexandra.weinberger@fernuni-hagen.de}}%
\author[6]{Alexander Wolff}
\affil[6]{Universit\"at W\"urzburg} %

\authorrunning{S.~Cornelsen, J.~Kratochv\'{\i}l, M.~M\"unch, G.~Ortali, A.~Weinberger, and A.~Wolff}

\newcommand{\eL}{\ensuremath{\mathsf{L}}}
\newcommand{\Le}{\scalebox{-1}[1]{\eL}}
\newcommand{\eLLe}{\eL-\Le}
\newcommand{\fe}{\mathrm{fe}}
\newcommand{\curve}{\eL}

\graphicspath{{./finalFigures/}}

\begin{document}

\maketitle

\begin{abstract}
  In a {\em grounded string representation} of a graph there is a horizontal
  line $\ell$ and each vertex is represented as a simple curve below
  $\ell$ with one end point on $\ell$ such that two curves intersect
  if and only if the respective vertices are adjacent. A grounded
  string representation is a {\em grounded
  \eLLe-representation} if each vertex is represented
  by a 1-bend orthogonal polyline. It is a {\em grounded \eL-representation}
  if in addition all curves are \eL-shaped.  We show that every
  biconnected series-parallel graph without edges between the two
  vertices of a separation pair (i.e., {\em transitive edges}) admits
  a grounded \eLLe-representation if and only if it
  admits a grounded string representation. Moreover, we can test in
  linear time whether such a representation exists. We also construct
  a biconnected series-parallel graph without transitive edges that
  admits a grounded \eLLe-representation, but no
  grounded \eL-representation.
\end{abstract}

\section{Introduction}

In this paper, we are interested in intersection representations of
graphs where each vertex is represented by an orthogonal polyline with one bend.
In particular, we consider the grounded case where all polylines start at a horizontal line and are drawn in the half-plane below that line. 
More precisely,  in a \emph{string representation}~\cite{ehrlichET76,sinden66} of a graph~$G$, each vertex~$v$ of~$G$ is represented by a simple curve $\curve(v)$ and the curves of two vertices $v$ and $w$ intersect if and only if $v$ and $w$ are adjacent. 
A string representation is a \emph{1-string representation} if any two curves intersect at most once.
In a \emph{grounded string representation} there is a horizontal line~$\ell$ such that the curve of any vertex $v$ is below $\ell$ and has one endpoint on $\ell$.  
This point is called the \emph{anchor} $\ell(v)$ of~$v$. 
A graph admits a grounded string representation if and only it admits an \emph{outerstring representation}~\cite{kratochvil:82} where each curve is drawn within a disk and the anchors are drawn on the boundary of this disk.
A graph is \emph{outerstring} if it has a grounded string representation and  \emph{outer-1-string} if it has a grounded string representation that is 1-string. 
While the computational complexity of recognizing outer-1-string graphs is open, recognizing  outerstring graphs is NP-hard~\cite{rokWalczak:2019}.

Curves with a low complexity are of particular interest. E.g., in \cite{antic_etal:gd25}, monotone curves with few bends and slopes $\pm 1$ were studied. 
These are special rotated orthogonal curves.
A string representation is \emph{orthogonal} if each curve is an orthogonal polyline, i.e., a sequence of vertical and horizontal segments.  An orthogonal string representation is an
\emph{\eLLe-representation} if each curve has at most one bend and this point is a bottommost point of the curve.  If a graph has a (grounded) \eLLe-representation, then it has a (grounded) string representation in which every curve is a straight-line segment \cite{mp-mcpcs-92}.
An orthogonal representation is an \emph{\eL-representation} if every curve has exactly one bend and this bend is the leftmost and bottommost point of the curve. 
Every planar graph admits an \eL-representation~\cite{goncalvesIP:soda18}.

Jel\'{\i}nek and T\"opfer~\cite{jelinek/toepfer:2019} discussed a hierarchy of grounded string representations. In particular, they showed that not every graph that admits a grounded \eLLe-representation also admits a grounded \eL-representation.
Outerplanar graphs,  interval graphs, and complete graphs always admit a grounded \eL-representation \cite{jelinek/toepfer:2019}.
Monotone \eL-representations, where each curve is anchored at its bend on a diagonal line, can always be found for outerplanar graphs, interval graphs,  2D-ray graphs, convex bipartite graphs, and complete graphs~\cite{ahmed:arxiv17,catanzaro_etal:dam17}.

A graph together with an ordering of the vertices is an \emph{ordered graph}.
In a \emph{constrained grounded string representation}~\cite{sinden66} of an ordered graph,
the anchors are in the order prescribed by the input order. Given an ordered graph with $n$ vertices and $m$ edges, it can be decided in $\mathcal O(n^4)$ time
whether it admits a constrained grounded \eL-representation~\cite{jelinek/toepfer:2019}, in $\mathcal O(m + n \log n)$ time whether it has a constrained monotone \eL-representation~\cite{ahmed:arxiv17}, and in
$\mathcal O(n^2)$ time whether it admits a constrained grounded \eLLe-representation~\cite{biedl_et_al:gd24}.
In general, one bend per curve (assuming orthogonal curves) 
does not suffice for ordered trees or cycles with a constrained grounded 1-string representation, but two bends per curve are always sufficient and such a representation can be computed in linear time~\cite{biedl_et_al:gd24}.

\subparagraph*{Our Results.}  
We first give a necessary condition for biconnected outerstring graphs in terms of components of a separation pair; see \cref{sec:necessary}. 
In \cref{thm:nonTransitiveFewHeavySuf}, we then demonstrate that these conditions even guarantee grounded \eLLe-representations for biconnected series-parallel graphs without transitive edges; where 
a biconnected graph is \emph{series-parallel} if and only if it does not contain a subdivision of the complete graph $K_4$ with four vertices~\cite{duffin:65}, and a \emph{transitive edge} is an edge between the two vertices of a separation pair. 
In this case, the necessary and sufficient condition can be tested in linear time; see \cref{cor:algo}.
Finally, we give an example of a biconnected series-parallel graph without transitive edges that admits a grounded \eLLe-representation, but no grounded \eL-representation; see \cref{thm:noL}.
\textit{For full proofs, we refer to the appendix.}

\renewcommand{\bottomfraction}{.99}

\section{A Necessary Condition for Biconnected Outerstring Graphs}\label{sec:necessary}

We give a necessary condition for biconnected outerstring graphs in terms of components of a separation pair. 
For a graph $G$ and a vertex~$s$ of~$G$, we let $G-s$ denote the graph
obtained from $G$ by removing $s$ and its incident edges.  Given a
biconnected graph~$G$, a \emph{separation pair} is a pair~$(s,t)$
of vertices of $G$ such that $G-s-t$ is disconnected. A
\emph{component}~$G'$ of $G$ w.r.t.\ $(s,t)$ is the subgraph induced
by $s$, $t$, and a connected component of $G-s-t$. The vertices $s$
and $t$ are the \emph{poles} of $G'$; the other vertices of $G'$ are
\emph{internal vertices}.
A component $G'$ w.r.t.\ a separation pair $(s,t)$ is \emph{heavy} if it contains an internal vertex that is neither adjacent to~$s$ nor to~$t$; otherwise $G'$ is \emph{light}.  A component~$G'$ is \emph{critical} if 
\begin{enumerate}
\item it contains an internal vertex $v$ that is neither adjacent
  to~$s$ nor to~$t$ (in particular $G'$ is heavy), and
\item there is a path $s-p_s-v-p_t-t$ in~$G'$ such that no internal
  vertex in the subpath $p_s$ is adjacent to $t$ and no internal
  vertex in the subpath $p_t$ is adjacent to $s$.
\end{enumerate}

A biconnected outer-1-string graph
can have many heavy components w.r.t.\ some separation pair; see \cref{fig:many}. However, only few of them can be critical.

\begin{figure}[tb]
	\begin{minipage}[t]{0.465\linewidth}
		\centering
		\includegraphics[page=1]{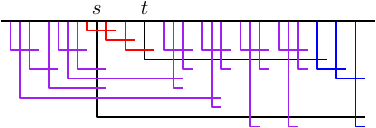}\\[1ex]
                \includegraphics[page=1]{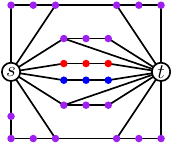}
		\subcaption{poles not adjacent}
		\label{fig:many_nontransitive}
	\end{minipage}
    \hfill
	\begin{minipage}[t]{0.5\linewidth}
		\centering
		\includegraphics[page=2]{heavy}\\[1ex]
                \includegraphics[page=2]{node-link.pdf}
		\subcaption{poles adjacent}
		\label{fig:many_transitive}
        \end{minipage}
	\caption{There can be many heavy components w.r.t.\
      a separation pair $(s,t)$ in a grounded \eL-representation; 
      the purple components are not critical.}
    \label{fig:many}
\end{figure}

\begin{lem}\label{lem:fewHeavyNec}
	Let $(s,t)$ be a separation pair of a biconnected outerstring graph $G$.
	\begin{enumerate}
	\item If $s$ and $t$ are not adjacent, then at most two components w.r.t.\ $(s,t)$ are critical. 
	\item If $s$ and $t$ are adjacent and the curves of $s$ and $t$ cross once, then at most four components w.r.t.\ $(s,t)$ are critical.
	\end{enumerate}
\end{lem}

\begin{figure}[htb]
    \centering
	\begin{minipage}[b]{0.2\linewidth}
		\centering
		\includegraphics[page=3]{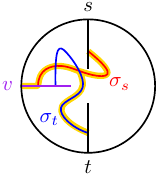}
		\subcaption{curves}
		\label{fig:boundCurves}
	\end{minipage}
    \hfil
	\begin{minipage}[b]{0.25\linewidth}
		\centering
		\includegraphics[page=4]{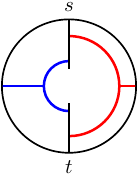}
		\subcaption{poles not adjacent}
		\label{fig:two_nontransitive}
    \end{minipage}
    \hfil
	\begin{minipage}[b]{0.2\linewidth}
          \centering
          \includegraphics[page=5]{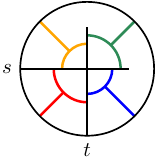}
          \subcaption{poles adjacent}
          \label{fig:four_transitive}
        \end{minipage}
	\caption{Illustration of \cref{lem:fewHeavyNec}.}
\end{figure}

\begin{proof}
    Consider an outerstring representation of $G$.
	Let $G'$ be a critical component w.r.t.\ $(s,t)$, and let $v$ be a vertex in $G'$ that is neither adjacent to $t$ nor to $s$, together with an $s$-$t$-path $s-p_s-v-p_t-t$ such that no vertex in $p_s$ is adjacent to $t$ and no vertex in $p_t$ is adjacent to $s$. In the representation of the vertices of~$p_s$, consider the parts between the intersection points between the curves of any two consecutive vertices. This yields a simple curve $\sigma_s$. Define $\sigma_t$ analogously for $p_t$. By the choice of $p_s$ and $p_t$,  $\sigma_t$ must not intersect the curve of $s$ and $\sigma_s$ must not intersect the curve of $t$; see \cref{fig:boundCurves}. 	
	The combination of $\sigma_s$, $\sigma_t$, and the representation of $v$ contains a 3-star connecting the curve of $t$, the curve of $s$, and the boundary at the anchor of $v$; see the thick yellow curve in \cref{fig:boundCurves}. If $s$ and $t$ are not adjacent, there can be at most two such 3-stars with the property that they are disjoint; see \cref{fig:two_nontransitive}. If the curves of $s$ and $t$ cross exactly once, then there can be at most four such 3-stars with the property that they are disjoint; see \cref{fig:four_transitive}.
\end{proof}

\section{Series-Parallel Graphs without Transitive Edges}\label{sec:noTransitive}

In this section, we show that the necessary condition in \cref{lem:fewHeavyNec} is also sufficient for biconnected series-parallel graphs without transitive edges and that in this case the condition can be tested in linear time. If there are no transitive edges, then the second condition in the definition of critical components is obsolete, and we obtain the following lemma.
\begin{restatable}[{\hyperref[lem:heavyIsCriticalA]{$\star$}}]{lem}{heavyIsCritical}
\label{lem:heavyIsCritical}
	A component of a biconnected series-parallel graph without transitive edges is critical if and only if it is heavy.
\end{restatable}
The separation pair $(s,t)$ is \emph{$k$-heavy} if it has $k$ heavy components. A graph~$G$ is \emph{$k$-heavy} if $k$ is the smallest integer such that every separation pair of $G$ is at most $k$-heavy.

\begin{theorem}\label{thm:nonTransitiveFewHeavySuf}
	If $G$ is a biconnected series-parallel graph without transitive edges, then the following statements are equivalent.
\begin{enumerate}
	\item\label{claim:notbad} $G$ is at most 2-heavy.
	\item\label{claim:outer-1} $G$ has an outerstring representation.
	\item\label{claim:groundedL} $G$ has a grounded \eLLe-representation.
\end{enumerate}
\end{theorem}
\begin{proof}[Proof Sketch](3) implies (2). If there are no transitive edges, then
  there is no edge between any two poles.  Thus, by
  Lemmas~\ref{lem:fewHeavyNec} and~\ref{lem:heavyIsCritical}, if $G$
  is outerstring, then $G$ is at most 2-heavy.  It remains to show that
  (\ref{claim:notbad}) implies (\ref{claim:groundedL}). %

  Let $(s,t)$ be a $k$-heavy separation pair with maximum $k$, and
  let~$G_1$ and~$G_2$ be two components of $(s,t)$ such that $G_1$ is
  heavy if $k \geq 1$, and $G_2$ is heavy if $k=2$.
For $i \in \{1,2\}$, let~$p_i$ be a longest $s$-$t$-path in~$G_i$. The paths $p_1$ and $p_2$ are induced paths (\cref{claim:induced}). 
The concatenation of $p_1$ and the reverse of $p_2$ is the \emph{heavy cycle} $c$. We draw $c$ as shown in \cref{fig:heavyCycle}.

\begin{figure}[b]
	\centering
	\begin{minipage}[b]{0.25\linewidth}
	\centering
	\includegraphics[page=6]{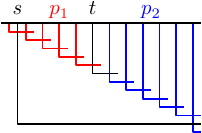}
	\subcaption{heavy cycle $c=p_1p_2$}
	\label{fig:heavyCycle}
        \end{minipage}
        \hfil
	\begin{minipage}[b]{0.36\linewidth}
	\centering
	\includegraphics[page=13]{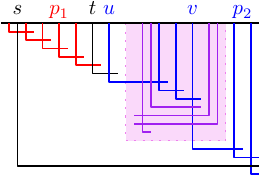}
	\subcaption{light component without $s$}
	\label{fig:LightwithoutS}
        \end{minipage}
        \hfil
	\begin{minipage}[b]{0.334\linewidth}
	\centering
	\includegraphics[page=14]{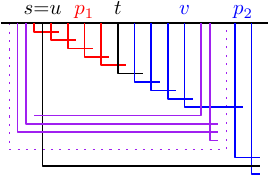}
	\subcaption{light component with $s$}
	\label{fig:lightWithS}
        \end{minipage}

        \bigskip
        
	\begin{minipage}[b]{\linewidth}
	\centering
	\includegraphics[page=32]{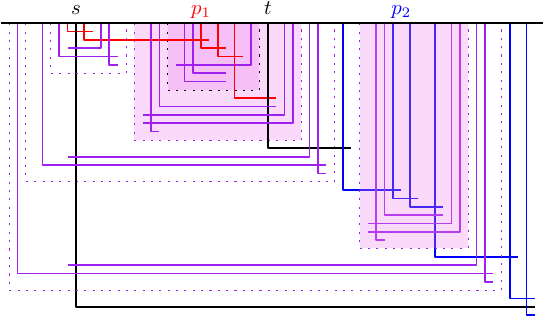}
	\subcaption{light components are disjoint or nested}
	\label{fig:heavyCyclewithLight}
        \end{minipage}

        \caption{How to construct a grounded
          \eLLe-representation for a
          series-parallel graph without transitive edges if there is
          one. The black \eL{}s are the poles, the red and the blue \eL{}s
          represent a longest cycle through the poles, and the purple
          curves represent light components.}
        \label{fig:nonTransitiveL}
\end{figure}

In order to also insert the remainder of the graph, consider a
connected component $G'$ of $G-p_1-p_2$.  Then $G'$ is connected to
exactly two vertices~$u$ and~$v$ on~$p_1$ or on~$p_2$
(\cref{claim:twoinpi}) 
and $G'$ together with~$u$ and~$v$ is light
w.r.t.\ $(u,v)$ (\cref{claim:nonheavy}).
Moreover, in a series-parallel graph without transitive edges, a light component $G'$ w.r.t.\ a separation pair $(u,v)$ looks as follows (\cref{claim:nonheavylooklike}): %
$G'$ either contains a single vertex adjacent to both~$u$ and~$v$, or $G'$ contains an independent set $V_u$ of vertices adjacent to $u$ and an independent set $V_v$ of vertices adjacent to~$v$ such that~$V_v$ and~$V_u$ are disjoint and~$V_v$ or~$V_u$ contains exactly one vertex.
Moreover, the separation pairs
corresponding to these light components are nested along
$p_1$ and~$p_2$ (\cref{claim:nested}).

Now we draw the  components of $G-p_1-p_2$ inside out as shown in \cref{fig:LightwithoutS,fig:lightWithS}. We start with the components that are not adjacent to $s$. Let $p = c-s$.
 For two vertices $u$ and $v$ on $p$, we write $u < v$ if the vertical part of $u$ is to the left of the vertical part of $v$.
Let $G'$ and $G''$ be two components of $G-p_1-p_2$, and let $u<v$ and $u''<v''$ be the pairs of vertices on $p$ that are  adjacent to $G'$ and $G''$, respectively. We process $G'$ before $G''$ if 
\begin{enumerate}
\item $G'$ is nested inside $G''$, i.e., if $u'' \leq u$ and $v < v''$, or
\item $G'$ is to the left of $G''$, i.e., $v \leq u''$.
\end{enumerate}
When drawing a component $G'$ with poles $u<v$ on $p$, we use the horizontal part of $u$ and the vertical part of $v$ to connect the curves of $G'$. See \cref{fig:LightwithoutS}.

If $G'$ contains a single vertex $w$ adjacent to $u$ and $v$, we anchor $w$ immediately to the right of the anchor of $u$ and draw the horizontal part of $\curve(w)$ immediately above the horizontal part of $\curve(v)$ until it intersects the vertical part of $v$.
Otherwise, we anchor the vertices adjacent to $v$ immediately to the left of the anchor of the successor of $v$ on $p$ and draw them as~\Le. 
Then we anchor the vertices adjacent to~$u$ immediately to the right of the anchor of $u$ and draw them as~\eL. 
In any case, we draw the horizontal parts of the curves immediately above the horizontal part of $v$. 
The horizontal parts of the vertices adjacent to $v$ must be long enough such that they intersect the vertical parts of the vertices adjacent to $u$. 
 
 In the end, we draw the components that are adjacent to $s$. See \cref{fig:lightWithS}. We draw a component $G'$ before a component $G''$ if $v < v''$ where $v$ and $v''$ are the vertices on $p$ other than $s$ incident to $G'$ and $G''$, respectively. Observe that $v$ and $v''$ cannot
be incident to $s$ since otherwise $p_1$ or $p_2$ would not have been longest paths. 

When processing a component $G'$ with poles $s$ and $v$ on $p$, we use the vertical part of $s$ and the horizontal part of $v$. Let $v'$ be the successor of $v$ on $p$. 
If $G'$ contains a vertex $w$ adjacent to $s$ and $v$, then we anchor~$w$ immediately to the left of the anchor of $v'$ and draw the horizontal part of the~\Le\ representing $w$ immediately above the horizontal part of $\curve(v)$ until it reaches the vertical part of $\curve(s)$.
All other vertices are drawn as \eL{}s. We first anchor the vertices adjacent to $s$ to  the left of any anchor and then the vertices adjacent to
$v$ immediately to the left of $v'$.  We draw the horizontal parts of the curves immediately above the horizontal part of $v'$. The horizontal parts of the vertices adjacent to $s$ must be long enough such that they intersect the vertical parts of the vertices adjacent  to $v$.
\end{proof}

\begin{restatable}[{\hyperref[cor:algoA]{$\star$}}]{cor}{algo}
  \label{cor:algo}
  It can be tested in linear time whether a biconnected
  series-parallel graph without transitive edges is an outerstring
  graph, an outer-1-string graph, admits a grounded string
  representation with segments, or admits a grounded
  \eLLe-representation.
\end{restatable}
\begin{proof}[Proof Sketch]
	Since a grounded \eLLe-representation implies a grounded representation with segments, which in turn implies that the graph is outerstring, the four properties are fulfilled if and only if the graph is at most 2-heavy. 
    By the special structure of light components (see \cref{claim:nonheavylooklike}), a component is heavy if and only if a longest path between the poles has length at least four. Proceeding first bottom-up and then top-down in the decomposition tree, we can compute the length of a longest path.
\end{proof}

\begin{figure}[htb]
    \centering
    \includegraphics[page = 1]{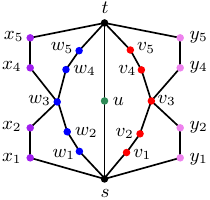} \hfil
    \includegraphics[page = 4]{noL.pdf}
    \caption{A series-parallel graph that admits a grounded \eLLe- but no grounded \eL-representation.}
    \label{fig:mirrored_needed}
\end{figure}

\begin{restatable}[{\hyperref[thm:noLA]{$\star$}}]{theorem}{noL}
\label{thm:noL}
    Graph $G$ in \cref{fig:mirrored_needed} is biconnected, series-parallel, and has no transitive edges. It admits a grounded \eLLe-representation but no grounded \eL-representation.
\end{restatable}
    
\begin{proof}[Proof Sketch]
  \cref{fig:mirrored_needed} depicts the
  \eLLe-representation.  In order to show
  that $G$ does not admit a grounded \eL-representation, we analyze
  three different cases, depending on the relative positions of
  $\curve(u)$, $\curve(s)$, and $\curve(t)$; see \cref{fig:3casesNoL}. 
  The curves $\curve(s)$, $\curve(t)$, and a part of $\curve(u)$ divide the area
  of the representation into a bounded (shaded) and an unbounded region. %
  We prove that
  $\{v_2,v_3,v_4\}$ has to be in one of them  and
  $\{w_2,w_3,w_4\}$ in the other.  For each case, we show
  that there are few ways to obtain a grounded \eL-representation of
  $G$ without the $x_i$ and $y_i$ and that for each of them, it is
  impossible to extend it to a representation of $G$.
  See \cref{fig:exampleNoLcompletion} for examples where $\curve(y_4)$ or
  $\curve(y_5)$ (left) and $\curve(x_1)$ or $\curve(x_2)$ (right) cannot be added.
\end{proof}

\begin{figure}[h]
        \centering
        \includegraphics[page = 23]{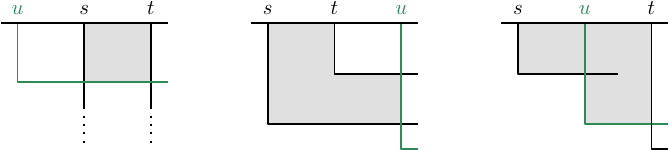}
        \caption{The three cases for the proof of \cref{thm:noL}.}
        \label{fig:3casesNoL}
\end{figure}

\begin{figure}[h]
        \centering
        \includegraphics{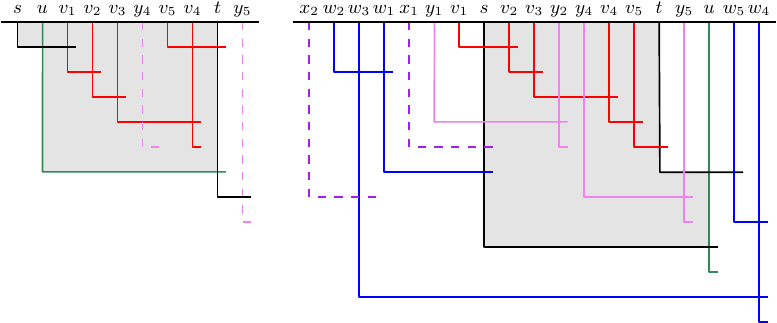}
        \caption{Representations that cannot be extended to a grounded \eL-representation of~$G$.}
        \label{fig:exampleNoLcompletion}
\end{figure}

\bibliographystyle{plainurl}
\bibliography{references}

\newpage

\appendix

\section{Missing Proofs for \cref{sec:noTransitive}}

\heavyIsCritical*\label{lem:heavyIsCriticalA}
\begin{proof}
	Let $G$ be a biconnected series-parallel graph without transitive edges and let $G'$ be a component with respect to some separation pair $(s,t)$.  It follows from the definition that every critical component is heavy. So assume now that $G'$ is heavy. Let $v$ be a vertex of $G'-s-t$ that is neither adjacent to $s$ nor $t$. Let $q$ be an $s$-$t$-path outside $G'$, let $e$ be the edge of $q$ incident to $s$, and let $e'$ be any edge incident to $v$. Since $G$ is biconnected, there must be a simple cycle $c$ containing $e$ and $e'$. Since $c$ contains an edge inside $G'$ and an edge outside $G'$, it follows that the cycle $c$ must also contain $t$  
    and, thus, an $s$-$t$-path through $G'$ containing $v$. Let $p$ be the part of $c$ in $G'$, let $p_s$ be the subpath of $p$ between $s$ and~$v$, and let $p_t$ be the subpath of $p$ between $v$ and~$t$. %
    We will show that no vertex of $p_s$ is adjacent to~$t$. Since the argument for $p_t$ is analogous, this implies that $G'$ is critical.
	
	Suppose that $p_s$ contains a vertex $w$ adjacent to $t$. Since there are no transitive edges, this implies that $(w,t)$ is not a separation pair. Thus, there is an $s$-$v$-path $p_s'$ that contains neither $w$ nor $t$.  Consider the intersection points of $p$ and $p_s'$ in the order in which they appear on $p_s'$. The first intersection point is $s$ which is before $w$ on $p$ and the last intersection point is $v$ which is after $w$ on $p$. So eventually there must be two consecutive intersection points  $u$ and $u'$ on $p_s'$ such that $u$ is before $w$ on $p$ and $u'$ is after $w$ on $p$.  Then $G'$ together with the $s$-$t$-path $q$ outside $G'$ contains a subdivision of a $K_4$ with vertices $u$, $u'$, $w$, and~$t$ (see \cref{fig:K4p1heavyIsCritical}), which yields the desired contradiction.
\end{proof}

\begin{figure}[htb]
	\begin{minipage}[b]{0.28\linewidth}
		\centering
		\includegraphics[page=7]{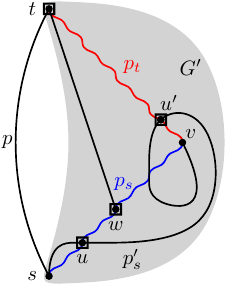}
		\subcaption{\dots{} \cref{lem:heavyIsCritical}}
		\label{fig:K4p1heavyIsCritical}
	\end{minipage}
        \hfill
	\begin{minipage}[b]{0.28\linewidth}
		\centering
		\includegraphics[page=8]{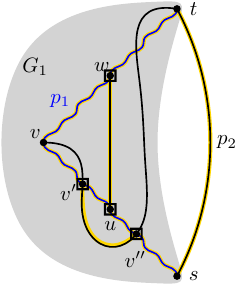}
		\subcaption{\dots{} \cref{claim:induced}}
		\label{fig:K4p1induced}
	\end{minipage}
        \hfill
	\begin{minipage}[b]{0.28\linewidth}
		\centering
		\includegraphics[page=9]{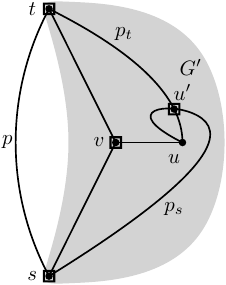}
		\subcaption{\dots{} \cref{claim:bipolarOne}}
		\label{fig:bipolarOne}
	\end{minipage}
	\caption{Subdivision of a $K_4$ used in the proof of \dots}
    \label{fig:K4}
\end{figure}

In the following, we provide the claims and their proofs needed for the proof of  \cref{thm:nonTransitiveFewHeavySuf}.

\begin{claim}\label{claim:induced}
	The paths $p_1$ and $p_2$ are induced paths.
\end{claim}
\begin{claimproof}
	Assume that $p_1= s,\dots,u,\dots,v,\dots,w,\dots,t$ where $u$ and $w$ are adjacent. Since there are no transitive edges, it follows that $(u,w)$ cannot be a separation pair. Hence, there is a path $p$ in $G$ between $v$ and $t$ that does not contain $u$ nor $w$. Since $G_1$ is a component of $G$, it follows that $p$ is completely in $G_1$. Let $v'$ be the last vertex of $p$ on $p_1$ that is between $u$ and $w$ and let $v''$ be the first vertex of $p$ after $v'$ that is on $p_1$. Then $G$ contains a subdivision of a $K_4$ consisting of the vertices $u$, $w$, $v'$, and $v''$ and the edges on $p_1$, $p_2$, the part of $p$ between $v'$ and $v''$, and the edge between $u$ and $w$. See \cref{fig:K4p1induced}. This contradicts the fact that $G$ is series-parallel.  
	
	The proof for $p_2$ is analogous.
\end{claimproof}

\begin{claim}\label{claim:twoinpi}
	For $i=1,2$, every connected component of $G_i-p_i$ is adjacent to exactly two vertices of $p_i$.
\end{claim}
\begin{claimproof}
	Let $G'$ be a connected component of $G_i-p_i$.  If $G'$ were adjacent to less than two vertices of~$p_i$, then the graph would not be biconnected. If $G'$ were connected to more vertices of $p_i$, then $G_i$ together with $p_{3-i}$ would contain a subdivision of $K_4$. 
\end{claimproof}
	
\begin{claim}\label{claim:nested}
	For $i=1,2$, let $G'$ and $G''$ be two connected components of $G_i-p_i$. 
    Then the neighbors of $G'$ and $G''$ on $p_i$ are nested, i.e., the order on $p_i$ is either $s$, a neighbor of $G'$, a neighbor of $G''$, a neighbor of $G''$, a neighbor of $G'$, and then $t$, or  $s$, a neighbor of $G''$, a neighbor of $G'$, a neighbor of $G'$, a neighbor of $G''$, and then $t$ (where consecutive vertices may be equal).
\end{claim}
\begin{claimproof}
	Otherwise, there is a subdivision of a $K_4$ contained in $G_i$ and $p_{3-i}$.
\end{claimproof}

The next claim implies that the heavy components are covered by $p_1$ and $p_2$.

\begin{claim}\label{claim:nonheavy}
	For $i=1,2$, every vertex in $G_i-p_i$ is adjacent to at least one vertex on $p_i$. %
\end{claim}

\begin{claimproof}
	Let $G'$ be a connected component of $G_i-p_i$. By \cref{claim:twoinpi}, $G'$ is adjacent to exactly two vertices $u, v$ on $p_i$. Assume that $G'$ contains a vertex $w$ that is not adjacent to any vertex on~$p_i$. 
	We first show that there is a $u$-$v$-path of length at least four whose internal vertices all belong to $G'$:
	Let $e$ be an edge incident to $w$ and let $e'$ be an edge of the subpath $p_{uv}$ of $p_i$ between $u$ and $v$. By biconnectivity, there must be a simple cycle containing $e$ and $e'$. This cycle~$c$ must contain $u$, $v$, $w$ and distinct neighbors $u'$ and $v'$ of $u$ and $v$, respectively, in $G'$. Thus, the internal vertices of the desired $u$-$v$-path are exactly the vertices on $c$ that belong to $G'$. 
		
	Observe that $(u,v)$ is a separation pair with at least three components, one containing~$G'$, one containing $p_{uv}$ and one containing $G_{3-i}$ (note that 
    $(u,v)$ separates $p_{uv}$ and $G_{3-i}$). Recall that if any separation pair had two heavy components, then we chose $G_1$ and $G_2$ to be the heavy components of $(s,t)$, i.e., in this case $G_{3-i}$ is heavy. Moreover, since $G'$ contains a vertex not adjacent to $p_i$ it follows that $G'$ is heavy by definition. Finally, since $p_i$ is a longest path and since $G'$ contains a $u$-$v$-path of length four, it follows that $p_{uv}$ has length at least four. Since $p_i$ is an induced path, it follows that $p_{uv}$ contains a vertex that is neither adjacent to $u$ nor $v$. It follows that the component containing $p_{uv}$ is a third heavy component w.r.t.\ $(u,v)$, a contradiction.
	So every vertex of $G'$ is adjacent to some vertex of~$p_i$.  %
\end{claimproof}

\begin{claim}\label{claim:bipolarOne}
	If a component 
	contains a vertex $v$ that is adjacent to both poles, then $v$ is the only internal vertex in this component.
\end{claim}
\begin{claimproof}
	Assume that there is a separation pair $(s,t)$ with a component $G'$ that contains a vertex $v$ that is connected to both $s$ and $t$ and to another vertex $u$. Since there are no transitive edges, neither $(s,v)$ nor $(v,t)$ can be a separation pair. Thus, there is a $u$-$s$-path $p_s$ that does not contain $v$ nor $t$ and a $u$-$t$-path $p_t$ that does not contain $v$ nor $s$.  Let $u'$ be the last vertex of $p_s$ on $p_t$.
	Let $p$ be an $s$-$t$-path through another component of $(s,t)$, i.e., $p$ contains no vertices of $G'$ except  $s$ and $t$. Then there is a subdivision of a $K_4$ with vertices $v,u',s,t$ and edges in $p,p_s,p_t$,  and the edges between $v$ and $s$, $u$, and $t$. See \cref{fig:bipolarOne}.
\end{claimproof}

\begin{claim}\label{claim:nonheavylooklike}
	The light components w.r.t.\ a separation pair $(s,t)$ look as follows:
	\begin{enumerate}
		\item a path of length two between the poles
		\item the poles $s$ and $t$, a vertex $t'$ incident to $t$ and an independent set $V_s$ of vertices incident to $t'$ and $s$. 
		\item the poles $s$ and $t$, a vertex $s'$ incident to $s$ and an independent set $V_t$ of vertices incident to $s'$ and $t$. 
	\end{enumerate}
\end{claim}

\begin{figure}[htb]
	\centering
	\begin{minipage}[b]{0.33\linewidth}
		\centering
		\includegraphics[page=10]{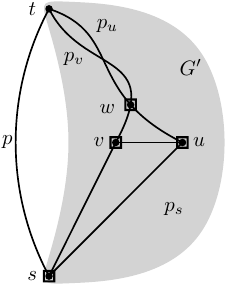}
	\end{minipage}\hfil
	\begin{minipage}[b]{0.3\linewidth}
		\centering
		\includegraphics[page=11]{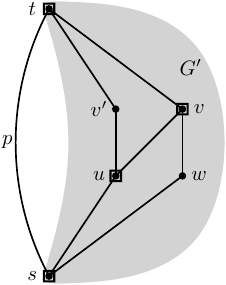}
	\end{minipage}\hfil
	\caption{Subdivision of a $K_4$ used in the proof of \cref{claim:nonheavylooklike}}\label{fig:K4-claim6}
\end{figure}

\begin{claimproof} 
	Let $G'$ be a light component of some separation pair $(s,t)$.
	If there is a vertex $v$ incident to $s$ and $t$, then, by \cref{claim:bipolarOne}, $v$ is the only vertex of $G'$ other than the poles. So, if $G'$ contains more than one internal vertex, then each vertex is adjacent to $s$ or $t$, but not to both. 
	
	Let $u$ and $v$ be two vertices of $G'$ that are both adjacent to $s$.  Assume that $u$ and $v$ are adjacent. Since there are no transitive edges, neither $(s,u)$ nor $(s,v)$ can be a separation pair. So, there must be a path $p_u$ from $u$ to $t$ that contains neither $s$ nor $v$ and a path $p_v$ from $v$ to $t$ that contains neither $s$ nor $u$. Let $w$ be the first common vertex of $p_u$ and $p_v$. Since $(s,t)$ is a separation pair, there must be an $s$-$t$-path $p$ through some other component than $G'$. Thus, there would be a subdivision of a $K_4$ with vertices $s$, $u$, $v$, $w$ within $G'$ and $p$. See \cref{fig:K4-claim6} on the left. Thus, all neighbors of $s$ in $G'$ form an independent set. Analogously, all neighbors of $t$ in $G'$ form an independent set. 
	
	Assume now that the set $V_s$ of neighbors of $s$ and the set $V_t$ of neighbors of $t$ in $G'$ both contain at least two vertices.  Since $G'-s-t$ is connected, there must be a tree $T$ connecting $V_s$ and $V_t$. Let $u \in V_s \cup V_t$ be a vertex that is not a leaf in $T$. Such a vertex exists since $|V_s \cup V_t| > 2$. The degree of $u$ in $T$ is at least 2. Assume w.l.o.g. that $u\in V_s$.
    Not all neighbors of $u$ can be leaves since $u$ would be the only vertex in $V_s$ otherwise. Let $v$ and $v'$ be neighbors of $u$ such that the degree of $v$ in $T$ is at least two and let $w \neq u$ be another neighbor of $v$ in $T$. Then there is a subdivision of a $K_4$ with vertices $u$, $v$, $s$ and $t$  %
    contained in $G'$ and an $s$-$t$-path $p$ outside $G'$, a contradiction. See \cref{fig:K4-claim6} on the right. 
\end{claimproof}

\subsection{A Series-Parallel Graph Admitting a Grounded L-\scalebox{-1}[1]{L}-representation but No Grounded L-representation.}

\noL* \label{thm:noLA}

\begin{figure}[t]
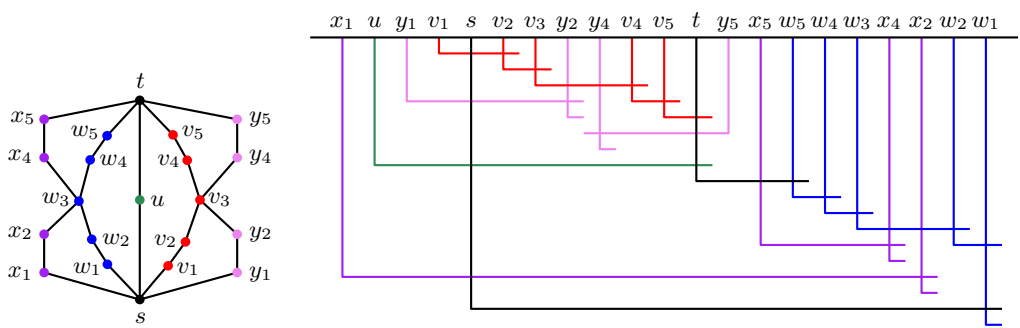

    \centering
    \includegraphics[page = 1]{noL.pdf} \hfil
    \includegraphics[page = 4]{noL.pdf}
    \caption{A series-parallel graph without transitive edges that admits a grounded \eLLe-representation but no grounded \eL-representation.}
    \label{fig:mirrored_needed2}
\end{figure}

In this section, we denote by $G$ the graph in \cref{fig:mirrored_needed2}. We prove that $G$ admits a grounded \eLLe-representation but it does not admit a grounded \eL-representation. The former is easy to observe in the figure. The section is devoted to proving the latter.  Let $X=\{x_1,x_2,x_4,x_5\}$, $Y=\{y_1,y_2,y_4,y_5\}$, $\Pi_w=\{w_1,w_2,w_3,w_4,w_5\}$, and $\Pi_v=\{v_1,v_2,v_3,v_4,v_5\}$.

Given a grounded 1-string representation of $G$, we denote by $\curve(v)$ the curve of a vertex~$v$. Here $\curve(v)$ does not have to be L-shaped. Let $\curve'(u)$ be the sub-curve of $\curve(u)$ from the point where $\curve(u)$ crosses $\curve(t)$ the point where it crosses $\curve(s)$.
Observe that $\curve' = \curve(u)'\cup \curve(s) \cup \curve(t)$ separates the area inside the drawing area (disk) into two regions, $R$ and~$\tilde{R}$.
Let $\curve'_v$ be the union of the curves $\curve(v_2),\dots,\curve(v_4)$ and the parts of $\curve(v_1)$ and $\curve(v_5)$ connecting $\curve(v_2)$ with $\curve(s)$ and $\curve(v_4)$ with $\curve(t)$. Let $\curve'_w$ be defined analogously. Since $\curve'_v$ and $\curve'_w$ intersects  $\curve'$ only in two end points, it follows that $\curve'_v$ and $\curve'_w$ are entirely contained in $R$ or $\tilde{R}$. Moreover, both connect $\curve(s)$ with $\curve(t)$ and the border. This implies that $\curve'_w$ has to be contained in $\tilde R$ if $\curve'_v$ is in $R$. We have the following lemma.

\begin{lem}
\label{le:splitinregions}
If $G$ admits a grounded 1-string representation, then it admits one where the curves of the vertices of $\Pi_v$ and $\Pi_w$ are inside $R$ and $\tilde{R}$, respectively, 
except for initial or final parts of the curves of $v_1, v_5, w_1,w_5$ (the vertices adjacent either to $s$ or $t$).
\end{lem}

\begin{figure}[t]
    \centering
    \includegraphics[page = 16,scale=0.8]{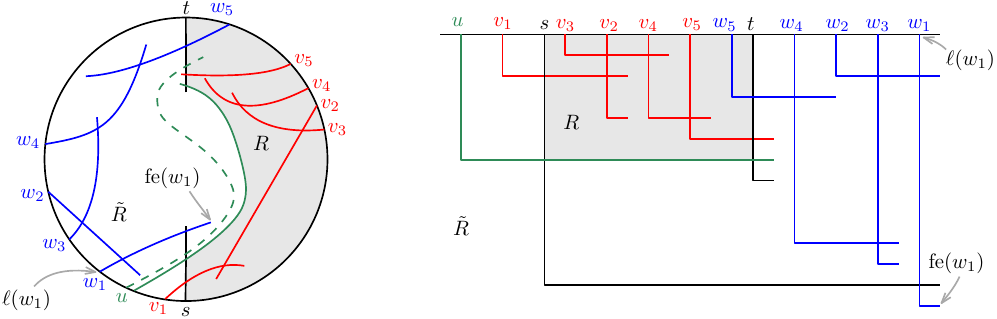}
    \caption{Illustration for the proof of \cref{le:splitinregions}.}
     \label{fig:circular}
\end{figure}

We define some notations. Given a grounded  \eL-representation,
the \emph{floating-end} of a vertex $v$, denoted by $\fe(v)$, is the point of $v$ that is the right endpoint of the horizontal line contained in $\curve(v)$. \cref{fig:circular} highlights $\ell(w_1)$ and $\fe(w_1)$. 

We denote by $X(\cdot)$ and $Y(\cdot)$ functions associating with a point in the plane its $X$- and $Y$-coordinate, respectively. The $Y$-coordinate below $\ell$ is considered positive, so for example in the right part of \cref{fig:circular} we have $Y(\fe(w_2))<Y(\fe(w_3))<Y(\fe(w_1))$. Observe that in a grounded \eL-representations we always have $X(\ell(v))<X(\fe(v))$.

In the rest of the section we mainly compare $X$-coordinates of anchors and $Y$-coordinates of floating-ends. Hence, in what follows, given two vertices $v$ and $v'$, for simplicity we write $\ell(v)<\ell(v')$ instead of $X(\ell(v))<X(\ell(v'))$ and  $\fe(v)<\fe(v')$ instead of $Y(\fe(v))<Y(\fe(v'))$. In \cref{fig:circular}, for example, we have $\ell(v_1)<\ell(s)$ and $\fe(v_2)<\fe(v_5)$. 

   \begin{figure}[htb]
    \centering
    \begin{minipage}[b]{0.3\linewidth}
		\centering
		\includegraphics[page = 5]{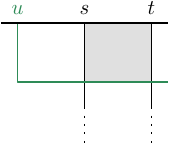}
		\subcaption{$\ell(u) < \ell(s) < \ell(t)$}
		\label{fig:poles-cases-a}
	\end{minipage}
    \hfil
    \begin{minipage}[b]{0.3\linewidth}
		\centering
		\includegraphics[page = 8]{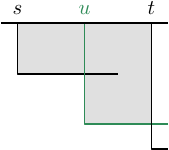}
		\subcaption{$\ell(s) < \ell(u) < \ell(t)$}
		\label{fig:poles-cases-b}
	\end{minipage}
    \hfil
    \begin{minipage}[b]{0.3\linewidth}
		\centering
		\includegraphics[page = 6]{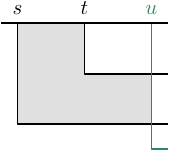}
		\subcaption{$\ell(s) < \ell(t) < \ell(u)$}
		\label{fig:poles-cases-c}
	\end{minipage}
    \hfil
    
    \caption{The cases analyzed separately in 
    \cref{lem:spnoL}.}
    \label{fig:poles-cases}
\end{figure}

  \begin{figure}[htb]
    \centering
    \begin{minipage}[b]{0.32\linewidth}
		\centering
		\includegraphics[page = 5]{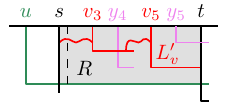}
		\subcaption{$\ell(u) < \ell(s) < \ell(t)$}
		\label{fig:ulowers}
    \end{minipage} \hfil
     \begin{minipage}[b]{0.33\linewidth}
		\centering
		\includegraphics[page = 6]{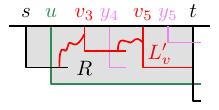}
		\subcaption{$\ell(s) < \ell(u) < \ell(t)$}
		\label{fig:ulowert}
    \end{minipage}\\[1em]
    \begin{minipage}[t]{0.28\linewidth}
		\centering
		\includegraphics[page = 10]{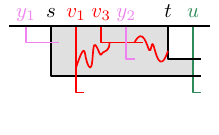}
		\subcaption{$\ell(s) < \ell(v_1)$ and \dots}
		\label{fig:slowerv1}
    \end{minipage} \hfil
    \begin{minipage}[t]{0.35\linewidth}
		\centering
		\includegraphics[page = 8]{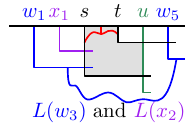}
		\subcaption{$\ell(w_1) < \ell(v_1) < \ell(s)$ and \dots}
		\label{fig:w1lowers}
    \end{minipage} \hfil
    \begin{minipage}[t]{0.35\linewidth}
		\centering
		\includegraphics[page = 9]{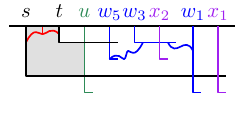}
		\subcaption{$\ell(v_1) < \ell(s) < \ell(w_1)$ and \dots}
		\label{fig:slowerw1}
    \end{minipage} \\[1em]
    \centerline{$\dots \ell(s) <\ell(t) < \ell(u).$}
    
    \caption{Illustrations for \cref{lem:spnoL}.}
    \label{fig:schematicCases}
\end{figure}

\begin{lem}\label{lem:spnoL}
    Graph $G$ has no grounded \eL-representation.
\end{lem}

\begin{proof}
   Suppose that $G$ admits a grounded \eL-representation.  The
   half-plane below the horizontal line $\ell$ is divided into two
   regions by the curve $\curve(s)$, the curve $\curve(t)$, and the part
   of the curve $\curve(u)$ between $\curve(s)$ and $\curve(t)$.  Let $R$ be the
   bounded region, and let $\tilde R$ be the unbounded region.

Observe that by \cref{le:splitinregions} we can assume that $\eL'_v$, i.e., the curves of the vertices of~$\Pi_v$, except initial or final parts of $v_1$ or $v_5$, are inside $R$ and that $\eL'_w$ is 
inside $\tilde R$. 
We assume w.l.o.g.\ that $\ell(s)<\ell(t)$, and we analyze the
possible positions of $\ell(u)$.

\smallskip \noindent
\textbf{Case} $\mathbf{\ell(u)<\ell(s)<\ell(t).}$ See \cref{fig:ulowers} for a schematic view. Since $\eL'_v \subset R$, we have that $\ell(s)<\ell(v_i)<\ell(t)$ ($i\in [2,4]$) implying $\ell(u)<\ell(v_1)<\ell(s)$  (otherwise $\curve(v_1)$ does not cross $\curve(s)$) and  $\ell(s)<\ell(v_5)<\ell(t)$. Moreover, $\ell(v_3) < \ell(v_5)$. Otherwise, since $\curve(v_5)$ has to intersect $\curve(t)$ but not $\curve(v_3)$, we get $\fe(v_5) > \fe(v_3)$. Since $\curve(v_2)$  has to intersect $\curve(v_3)$ but not $\curve(v_5)$, we would obtain $\ell(v_5) < \ell(v_2) < \ell(v_3)$ and $\fe(v_5) > \fe(v_2)$. But now, $\curve(v_1)$ cannot intersect $\curve(s)$ and $\curve(v_2)$. 

Consider now $y_4$ and $y_5$. 
We have that $\curve(y_4) \subset R$ since $\curve(y_4)$ crosses $\curve(v_3)$ but not~$\eL'_v$. 
$\curve(y_5)$ has to cross $\curve(t)$ and $\curve(y_4)$ without crossing $\curve(s)$. This implies in particular $\ell(s) < \ell(y_5) < \ell(t)$. On the other hand $\curve(y_5)$ must not cross  $\curve'_v$. Since 
any vertical line between $\ell(s)$ and $\ell(t)$ crosses $\curve'_v$ (see the dashed vertical line in the figure), this implies that $\ell(v_5) < \ell(y_5) < \ell(t)$. But since $\ell(v_3) < \ell(v_5)$ and $\curve(y_4)$ intersects $\curve(v_3)$ this also implies that $\curve(y_4)$ cannot intersect $\curve(y_5)$. 

\smallskip \noindent
\textbf{Case} $\mathbf{\ell(s)<\ell(u)<\ell(t).}$ See \cref{fig:ulowert} for a schematic view. %
This case is very similar to the previous case. 

Since $\curve'_v \subset R$ and since $\curve'_v \cap \curve(u) = \emptyset$ it follows that  $\ell(u)<\ell(v_i)<\ell(t)$  for $i\in [1,5]$. Since $\curve(y_4)$ has to cross $\curve(v_3)$ but not $L_v'$, we obtain again $\curve(y_4) \subset R$ and, thus, $\ell(u) < \ell(y_4) < \ell(t)$. 
For the same reasons as in the previous case it is now impossible that $\curve(y_5)$ crosses $\curve(t)$ and $\curve(y_4)$ without crossing $\curve(s)$ or $\curve'_v$.

\smallskip \noindent
\textbf{Case} $\mathbf{\ell(s)<\ell(t)<\ell(u).}$
    Since $\eL'_v$ contains a curve connecting $\curve(s)$ and $\curve(t)$ within $R$, we obtain again that $\ell(s)<\ell(v_i)<\ell(t)$ for $i\in [2,4]$.  We distinguish whether $\ell(v_1)<\ell(s)$ or $\ell(s)<\ell(v_1)$.

Consider the case $\ell(s)<\ell(v_1)$. See \cref{fig:slowerv1} for a schematic view. We have that $\curve(y_2)$ has to cross $\curve(s)$ or $\eL'_v$ in order to cross $\curve(y_1)$, similarly to the arguments for $\curve(y_5)$ and $\curve(y_4)$ in the previous two cases. 

Suppose now $\ell(v_1)<\ell(s)$. %
We have $\ell(u)<\ell(w_5)$ and $\fe(t)<\fe(w_5)$, otherwise $\curve(w_5)$ cannot cross $\curve(t)$. 

We first prove that  $\ell(w_1)<\ell(s)$ is not possible. See \cref{fig:w1lowers} for a schematic view. Otherwise, there is $p\in [2,4]$ such that $\ell(w_p)<\ell(s)$ and $\ell(w_5)\le \ell(w_{p+1})$ (where the equality holds if $p=4$). %
Then, we have that $\ell(w_1)<\ell(x_1)<\ell(s)$ otherwise $\curve(x_1)$ cannot cross $\curve(s)$ without crossing $\curve(w_i)$ for some $i\in [1,5]$. However, $\ell(w_1)<\ell(x_3)<\ell(s)$ is not possible otherwise $\curve(x_3)$ does not cross $\curve(w_3)$ without crossing $\curve(w_1)$. Hence, $\curve(x_1)$ and $\curve(x_3)$ cannot cross. 

So $\ell(w_1)>\ell(s)$. See \cref{fig:slowerw1} for a schematic view.
Since $\curve(w_1)$ has to intersect $\curve(s)$ without intersecting $\eL'_v$, $\curve(t)$, nor $\curve(u)$ and since $\fe(s)<\fe(u)$, it follows that $\ell(u)<\ell(w_1)$ and $\fe(s)<\fe(w_1)$. Also, $\ell(w_5)<\ell(w_1)$, otherwise $\curve(w_5)$ crosses $\curve(s)$. 
 
Since $\eL'_w$ connects $\curve(t)$ and $\curve(s)$, it follows that $\ell(u)<\ell(w_3) < \ell(w_1)$. 
Now consider $x_1$ and $x_2$. Since $\curve(x_1)$ has to cross $\curve(s)$ without crossing $\eL'_v$, we have $\ell(w_2)<\ell(x_1)$. Since $\curve(x_2)$ has to cross only $\curve(w_3)$ and $\curve(x_1)$, we have $\ell(x_2)<\ell(w_1)$. But then $\curve(x_2)$ and $\curve(x_1)$ cannot cross.
\end{proof}

\section{A Linear-Time Testing Algorithm}

Here, we provide the details for testing whether a series-parallel graph without transitive edges is outerstring and, thus, admits a grounded \eLLe-representation. We use the \emph{decomposition tree} of a series-parallel graph, a vertex-labeled ordered rooted tree $T$ defined as follows. The leaves of a decomposition tree are labeled $Q$, and all other vertices are labeled $S$ or $P$ such that no two adjacent vertices have the same label. 
The decomposition tree of a biconnected series-parallel graph can be constructed as follows.
If $G$ does not contain a separation pair, then $G$ is a triangle (Otherwise $G$ is 3-connected and, thus, contains a subdivision of a~$K_4$). In this case, consider any two vertices of $G$ as the poles $s$ and $t$.
Otherwise, start with an arbitrary separation pair $(s,t)$ and let $G_1,\dots,G_k$ be the components w.r.t.\ $(s,t)$. The poles of each $G_1,\dots,G_k$ are $s$ and $t$. The root of the decomposition tree is a P-node with $k$ children. If $s$ and $t$ are connected, then one child is a Q-node. All (other) children are S-nodes. The S-node representing $G_i$ for some $i \in \{1,\dots,k\}$ has a child for each biconnected component of~$G_i$ in the order in which they appear in $G_i$. The nodes are labeled Q or P, depending on whether $G_i$ is a single edge or not. The poles are $s$, $t$, and the cut vertices of~$G_i$.  Now continue recursively with the biconnected graphs, which correspond to P-nodes.   

\algo*\label{cor:algoA}
\begin{proof}
	Since a grounded \eLLe-representation implies a grounded representation with segments which in turn implies that the graph is outerstring, the four properties are fulfilled if and only if the graph is at most 2-heavy. By \cref{claim:nonheavylooklike}, a component is heavy if and only if a longest path between the poles has length at least four. Proceeding first bottom-up and then top-down in the decomposition tree, we can compute the length of a longest path as follows. 
	
	Let $G$ be a biconnected series-parallel graph, and let $T$ be its decomposition tree. We assume that $T$ is rooted at a Q-node.
	Bottom-up, we compute for every node $\nu$ of $T$ the length $\ell(\nu)$ of a longest path between the poles of the graph $G(\nu)$ represented by the subtree $T(\nu)$ of $T$ rooted $\nu$. If $\nu$ is a Q-node, then $\ell(\nu)=1$. Assume now that $\nu$ has children $\mu_1,\dots,\mu_k$. If $\nu$ is an S-node, then $\ell(\nu) = \sum_{i=1}^k\ell(\mu_i)$ and if $\nu$ is a P-node, then $\ell(\nu) = \max_{i=1}^k\ell(\mu_i)$.
	
	Next, we compute top down for each node $\nu$ the length $\ell'(\nu)$ of a longest path in the parent component $G - G(\nu)$ of $\nu$, i.e., the graph obtained from $G$ by removing all edges and vertices of $G(\nu)$ except the poles of $G(\nu)$.  
	If $\nu$ is the only child of the root, then $\ell'(\nu)=1$. Let now $\nu$ by some inner node of $T$ and assume that we have already computed $\ell'(\nu)$. Let $\mu_i$, $i=1,\dots,k$ be the children of $\nu$. If $\nu$ is an S-node, then $\ell'(\mu_i) = \ell(\nu) + \ell'(\nu) - \ell(\mu_i)$. If $\nu$ is a P-node, then let $\mu_j$ be a child of $\nu$ for some $j \in \{1,\dots,k\}$ with $\ell(\nu) = \ell(\mu_j)$. If $i \neq j$, then set $\ell'(\mu_i) = \max \{ \ell'(\nu),\ell(\nu) \}$. Otherwise,  $\ell'(\mu_i) = \max \{ \ell'(\nu),\max_{j' \in \{1,\dots,k\} \setminus \{j\}}\ell(\mu_{j'}) \}$.
	
	Let now $(s,t)$ be a separation pair of $G$ that has more than two components. Then $s$ and $t$ are the poles of a P-node. So for each P-node $\nu$ with children $\mu_1,\dots,\mu_k$ we have to check whether at most two among $\ell(\mu_1),\dots,\ell(\mu_k),\ell'(\nu)$ are greater than three. This is linear in the sum of the degrees in $T$ and, thus, linear in the size of $G$.
\end{proof}
\end{document}